\documentclass[letterpaper, 10 pt, conference]{ieeeconf}
\IEEEoverridecommandlockouts
\usepackage{booktabs}
\usepackage{makecell}
\usepackage{subfigure}
\usepackage{cite}
\usepackage{amsmath,amssymb,amsfonts}
\usepackage{bm}
\usepackage{optidef}
\usepackage{color}
\usepackage{mathrsfs}
\usepackage{enumerate}
\usepackage{textcomp}
\usepackage{stfloats}
\usepackage{url}
\usepackage{verbatim}
\usepackage{graphicx}
\usepackage{balance}

\usepackage[english]{babel}
\usepackage{amsthm}
\newtheorem{theorem}{Theorem}[section]
\newtheorem{definition}{Definition}[section]

\newtheorem{proposition}{Proposition}[section]
\newtheorem{remark}{Remark}[section]
\newtheorem{example}{Example}[section]
\newtheorem{corollary}{Corollary}[section]
\makeatletter
\let\NAT@parse\undefined
\makeatother
\usepackage[citebordercolor={0 1 0},linkbordercolor={0 1 1}]{hyperref}
\makeatletter
\let\oldhyper@linkurl\hyper@linkurl
\def\hyper@linkurl#1#2{%
\begingroup
\hypersetup{pdfborder={0 0 0}}%
\oldhyper@linkurl{#1}{#2}%
\endgroup}
\makeatother
\title{\LARGE \bf On the Degree of Safety: Beyond Safe or Unsafe with \\Control Barrier Functions}
\author{
Ruoyu Lin$^{1}$,
Fabio Pasqualetti$^{1}$,
and Magnus Egerstedt$^{2}$
\thanks{This work was supported in part by the U.S. Army Research Lab through ARL DCIST CRA W911NF-17-2-0181, and in part by the U.S. National Science Foundation under award CMMI-2622263.}
\thanks{$^{1}$Ruoyu Lin and Fabio Pasqualetti are with the Department of Electrical Engineering and Computer Science, University of California, Irvine, Irvine, CA 92697 USA. Email: {\tt\small \{\href{mailto:rlin10@uci.edu}{\tt\small rlin10}, \href{mailto:fabiopas@uci.edu}{\tt\small fabiopas\}@uci.edu}}}
\thanks{$^{2}$Magnus Egerstedt is with the University of North Carolina at Chapel Hill, Chapel Hill, NC 27599 USA. Email: {\tt\small \href{mailto:magnus@unc.edu}{\tt\small magnus@unc.edu}}}
}
\begin{document}
\maketitle
\thispagestyle{empty}
\pagestyle{empty}
\begin{abstract}
A valid control barrier function (CBF) certifies if its represented safe set can be rendered forward invariant, and the sign of its value indicates whether a state is safe or not, but it does not quantify a degree of safety beyond the binary indication. In this paper, we show that among valid CBFs representing the same safe set, interior values and gradients can be changed arbitrarily, so neither quantity determines a degree of safety that is independent of how the set is represented. We also show that whether a candidate CBF-based inequality constraint is feasible does not by itself quantify a degree of safety. In particular, infeasibility can occur either because the safe set is not controlled invariant or because the candidate CBF representation fails. This motivates our distinction between intrinsic and representational infeasibility. Finally, we introduce the invariance authority demand (IAD), a representation-independent degree of safety that quantifies the control authority required for maintaining safety and can be used to guide actuator redesign or safe set repair.
\end{abstract}

\section{Introduction} \label{sec:intro}
Safety of a dynamical system depends jointly on the system dynamics, the admissible control inputs, and the safe set prescribed by a specific task. A fundamental challenge is to verify if the safe set is controlled invariant \cite{blanchini1999set,mitchell2005time,fiacchini2010computation}. After a scalar function representing the safe set is verified to be a valid control barrier function (CBF), it provides a convenient way for control synthesis, typically through a CBF-based optimization problem (CBF-OP) \cite{ames2016control}. This paper asks a different question: Beyond the binary distinction between safe and unsafe, how should the \emph{degree of safety} be quantified?

One approach could be to associate a larger CBF value, or perhaps a larger CBF gradient norm, with a higher degree of safety. However, the sign of a CBF only determines if a state is in the safe set, and its value and gradient depend on a specific CBF selected to represent the same set, as illustrated in Fig.~\ref{fig:Illustration}. Another way could be to examine whether there exists an admissible control input satisfying the candidate CBF-based inequality, i.e., whether the candidate CBF-OP is feasible. However, such feasibility indicates neither if a state is in the safe set nor if the safe set is controlled invariant. For instance, a candidate CBF-OP can be feasible at states outside the safe set (see, e.g., \cite{ames2016control}) and infeasible at states in the safe set (see Example~\ref{ex:repair_representation}). In addition, it can be feasible for some intervals when the safe set is not controlled invariant (see Example~\ref{ex:repair_intrinsic}) and infeasible even when the safe set is controlled invariant (see Example~\ref{ex:repair_representation}).

A related issue is that simply plugging a scalar function into a candidate CBF-OP to solve for a controller without verifying whether such a scalar function is a valid CBF or not can lead to unjustified claims on safety, as also highlighted in \cite{kim2026your}. Indeed, such verification can be challenging in general and remains an active research area (see, e.g., \cite{robey2020learning,choi2021robust,clark2024semialgebraic}). This paper does not seek to resolve this challenge. Instead, our main focus is on the degree of safety, what information candidate CBF values, gradients, and feasibility actually provide about it, and how different types of infeasibility should be diagnosed and repaired.
\begin{figure}[t]
\centering
\includegraphics[scale=0.18]{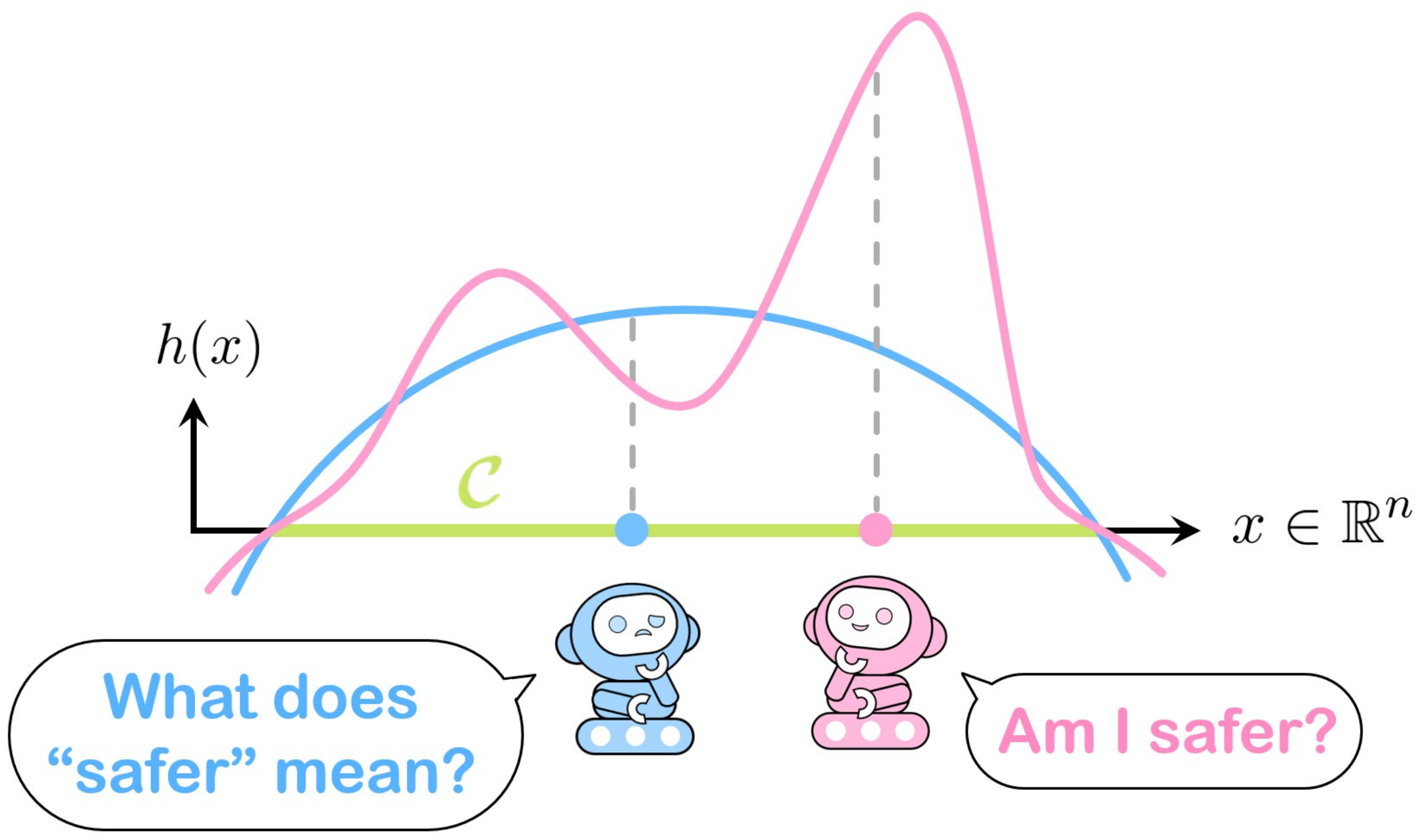}
\caption{Illustration of blue and pink CBFs representing the same safe set $\mathcal{C}$, i.e., having the same $0$-superlevel set. Under the pink CBF, the blue robot's state has smaller $h(x)$ and $\|\nabla{h}(x)\|$ but is farther from $\partial \mathcal{C}$ than the pink robot's state. Under the blue CBF, $h(x)$ is larger at the blue robot's state.}
\label{fig:Illustration}
\end{figure}

The main contributions of this paper are as follows.
\begin{enumerate}
\item We prove that for a valid CBF representing a safe set, its values and gradients can be arbitrarily modified at any finite number of interior states, after which it remains a valid CBF representing the same set. Therefore, CBF values and gradients cannot by themselves quantify a representation-independent degree of safety.
\item We show that the feasibility of a candidate CBF-based inequality constraint is by itself inconclusive, and introduce a two-stage diagnosis that distinguishes intrinsic from representational infeasibility, which identifies whether the repair should modify the safe set, the CBF representation, or the underlying controlled system.
\item We introduce the invariance authority demand (IAD), a dimensionless, representation-independent quantity that characterizes controlled invariance, quantifies a control authority-based degree of safety, and guides controlled invariant set construction.
\end{enumerate}

\section{Preliminaries} \label{sec:preliminaries}
Consider the dynamical system
\begin{equation}
\label{eqn:controlaffine}
\dot{x} = F(x,u)\coloneqq f(x) + g(x)u,
\end{equation}
where $x\in \mathcal{D} \subseteq \mathbb{R}^n$ and $u \in \mathcal{U} \subseteq \mathbb{R}^m$ are the state and control input, respectively, with $\mathcal{D}$ open and $\mathcal{U}$ convex and compact, and $f: \mathcal{D} \to \mathbb{R}^n$ and $g: \mathcal{D} \to \mathbb{R}^{n\times m}$ are locally Lipschitz vector fields.
\begin{definition}[Controlled Invariance]
\label{def:controlledinvariance}
A set $\mathcal{A} \subseteq \mathcal{D}$ is controlled invariant if, for any $x(0) \in \mathcal{A}$, there exists an admissible control input $u$ taking values in $\mathcal{U}$ such that $x(t) \in \mathcal{A}$, $\forall t \in[0,T_{\max})$, where $[0,T_{\max})$ is the maximum interval of existence of the trajectory of \eqref{eqn:controlaffine}.
\end{definition}
\begin{definition}[Forward Invariance]
\label{def:forwardinvariance}
A set $\mathcal{A} \subseteq \mathcal{D}$ is forward invariant under $u(x)$ if, for any $x(0)\in \mathcal{A}$, $x(t)\in \mathcal{A}$, $\forall t \in [0,T_{\max})$, where $[0,T_{\max})$ is the maximum interval of existence of the trajectory of \eqref{eqn:controlaffine}.
\end{definition}

Note that controlled invariance asks whether there exists a controller such that the state stays in a set, while forward invariance is a property of the closed-loop system under a particular feedback controller.

Let the safe set determined by a specific task be represented by a $C^1$ function $h:\mathcal{D} \to \mathbb{R}$ as
\begin{equation}
\label{eqn:safeset}
\mathcal{C}
\coloneqq 
\{x \in \mathcal{D} \mid h(x) \geq 0\},
\end{equation}
with $\partial \mathcal{C} \coloneqq  \{x \in \mathcal{D} \,|\, h(x) = 0\}$, $\operatorname{Int}(\mathcal{C}) \coloneqq \{x \in \mathcal{D} \,|\, h(x) > 0\} $, and $\nabla{h}(x) \neq 0$, $\forall x \in \partial\mathcal{C}$. Any such $h$ is called a representation of $\mathcal{C}$, and a state $x$ is said to be safe if $x \in \mathcal{C}$. Given an extended class $\mathcal{K}_{\infty}$ function $\alpha: \mathbb{R} \to \mathbb{R}$, we denote
\begin{equation}
\label{eqn:mu_halpha}
\mu_{h,\alpha}(x)
\coloneqq
\sup_{u \in \mathcal{U}}
\left(
\nabla{h}(x)^\top F(x,u) + \alpha(h(x))
\right),
\end{equation}
for any $x \in \mathcal{D}$. Intuitively, $\mu_{h,\alpha}(x) \geq 0$ means that some $u \in \mathcal{U}$ can prevent $h$ from decreasing faster than $-\alpha(h)$ at $x\in \mathcal{D}$. Then, using the notation of \eqref{eqn:mu_halpha}, the definition of CBF is presented below \cite{ames2016control}.
\begin{definition}[Control Barrier Function (CBF)]
\label{def:CBF}
A $C^1$ function $h$ is a CBF on $\mathcal{D}$ if there exists an extended class $\mathcal{K}_{\infty}$ function $\alpha$ such that $\mu_{h,\alpha}(x) \geq 0$, for all $x \in \mathcal{D}$.
\end{definition}
\begin{theorem}
\label{thm:CBFinvariant}
If $h$ is a valid CBF per Definition~\ref{def:CBF}, then any locally Lipschitz controller $u(x)$ satisfying the CBF-based inequality $\nabla{h}(x)^\top F(x,u(x)) + \alpha(h(x)) \geq 0$ renders the safe set $\mathcal{C}$ forward invariant.
\end{theorem}

After verifying a $C^1$ function $h$ is a valid CBF per Definition~\ref{def:CBF}, Theorem~\ref{thm:CBFinvariant} suggests an efficient way of synthesizing controllers via a CBF-OP \cite{ames2019control}, e.g.,
\begin{equation}
\label{eqn:CBFOP}
\begin{aligned}
&\underset{u\in \mathcal{U}}{\operatorname{arg\, min}}
\quad
\|u- u_{\operatorname{nom}}\|^2
\\
&\quad\quad\text{s.t.}
\quad
\nabla{h}(x)^\top F(x,u)
\geq
-\alpha(h(x)),
\end{aligned}
\end{equation}
where $u_{\operatorname{nom}} \in \mathbb{R}^m$ is the nominal control input. 

If the $C^1$ function $h$ used in \eqref{eqn:CBFOP} is only a candidate CBF that has not been verified to satisfy Definition~\ref{def:CBF} for all $x \in \mathcal{D}$, we refer to \eqref{eqn:CBFOP} as a candidate CBF-OP and its constraint as a candidate CBF-based inequality constraint.

For a selected pair of $h$ and $\alpha$, we denote $J_{h,\alpha} \coloneqq \inf_{x\in\mathcal{C}} \mu_{h,\alpha}(x)$. Then, $J_{h,\alpha} \geq 0$ implies that the selected scalar function $h$, with $\alpha$, is a valid CBF on $\mathcal{C}$. If $J_{h,\alpha} < 0$, the selected pair of $h$ and $\alpha$ results in infeasibility of the corresponding candidate CBF-OP somewhere in $\mathcal{C}$, which, however, does not imply that $\mathcal{C}$ is not controlled invariant.

To characterize controlled invariance, define
\begin{equation}
\label{eqn:mu_C}
\mu_{\mathcal{C}}(x)
\coloneqq
\sup_{u\in \mathcal{U}}
n_{\mathcal{C}}(x)^\top F(x,u),
\end{equation}
for any $x \in \partial\mathcal{C}$, where $n_{\mathcal{C}}(x) \coloneqq \nabla{h}(x)/\|\nabla{h}(x)\|$, and
\begin{equation}
\label{eqn:J_C}
J_{\mathcal{C}}
\coloneqq
\inf_{x \in \partial\mathcal{C}}
\mu_{\mathcal{C}}(x).
\end{equation}

\noindent Intuitively, $\mu_{\mathcal{C}}(x) \geq 0$ means that some $u \in \mathcal{U}$ can prevent the state from moving outside $\mathcal{C}$ at $x \in \partial\mathcal{C}$. Requiring this over the entire boundary, i.e., $J_{\mathcal C}\geq 0$, is equivalent to controlled invariance by Nagumo's theorem below \cite{nagumo1942lage}.
\begin{theorem}
\label{thm:Nagumo}
$\mathcal{C}\text{ is controlled invariant}
\;\Longleftrightarrow\;
J_{\mathcal{C}} \geq 0$.
\end{theorem}

With these preliminaries, we next investigate the question of how to quantify the degree of safety.

\section{On the Degree of Safety}
\subsection{What CBF Values and Gradients Do Not Tell} \label{sec:CBFgraddonottell}
Within Section~\ref{sec:CBFgraddonottell}, $h$ is assumed to be a valid CBF representing the safe set $\mathcal{C}$ per Definition~\ref{def:CBF}. As discussed in Section~\ref{sec:intro}, using CBF values as quantitative measures may suggest that larger values correspond to safer states. However, what ``safer'' means is not clear yet. One possible interpretation implicitly combines the following two implications. For any $x_A, x_B \in \mathcal{C}$,
\begin{equation}
\begin{aligned}
h(x_A) > h(x_B)
\;&\Longrightarrow\;
d_{\partial \mathcal{C}}(x_A) > d_{\partial \mathcal{C}}(x_B),\\
d_{\partial \mathcal{C}}(x_A) > d_{\partial \mathcal{C}}(x_B)
\;&\Longrightarrow\;
x_A \text{ is safer than } x_B,
\end{aligned}
\label{eqn:implicitassume}
\end{equation}
where $d_{\partial\mathcal C}(x)$ denotes the distance of $x$ to the safe set boundary ${\partial \mathcal{C}}$. However, the first implication does not necessarily hold, as demonstrated by the pink curve in Fig.~\ref{fig:Illustration}. The second implication is not meaningful until the intended notion of the degree of safety is specified.

Distance to the boundary can measure a statewise degree of safety, which does not depend on the CBF representation of the safe set, although it does not account for the system dynamics and input constraints. One may ask if the gradient of a CBF captures the missing information of system dynamics because $\nabla{h}$ appears in the CBF-based inequality constraint and can thus affect the closed-loop behavior. However, $\nabla{h}$ also depends on the representation of the set.
\begin{theorem}
\label{thm:interiorvalues}
Let $x_1, \ldots, x_N \in \operatorname{Int}(\mathcal{C})$ be distinct. For any collection of $c_i>0$, $w_i \in \mathbb{R}^n$, $i \in \mathcal{N} \coloneqq \{1,\ldots,N\}$, there exist a $C^1$ function $s: \mathcal{D} \to (0,\infty)$ and an open neighborhood of $\partial \mathcal{C}$, denoted as $\mathcal{S}$, such that $\widetilde{h}(x) \coloneqq s(x)h(x)$ with $\widetilde{h}(x) = h(x)$, $\forall x \in \mathcal{S}$, is also a valid CBF representing the same set $\mathcal{C}$, and
\begin{equation*}
\widetilde{h}(x_i) = c_i,
\;\;
\nabla{\widetilde{h}}(x_i) = w_i,
\;\;
\forall i \in \mathcal{N}.
\end{equation*}
\end{theorem}
\begin{proof}
Since $x_1,\ldots,x_N \in \operatorname{Int}(\mathcal{C})$ are distinct, then $\exists \, \epsilon_i > 0$ such that $\overline{\mathcal{V}}_i \subset \mathcal{B}_i$ and $\overline{\mathcal{B}}_i \subset \operatorname{Int}(\mathcal{C})$, $\forall i \in \mathcal{N}$, and $\overline{\mathcal{B}}_i \cap \overline{\mathcal{B}}_j = \varnothing$, $\forall i \neq j \in \mathcal N$, where $\overline{\mathcal{V}}_i$ and $\overline{\mathcal{B}}_i$ denote the closures of $\mathcal{V}_i \coloneqq \left\{x \in \mathcal{D} \mid \|x - x_i\| < \epsilon_i \right\}$ and $\mathcal{B}_i \coloneqq \left\{x \in \mathcal{D} \mid \|x - x_i\| < \varepsilon \epsilon_i \right\}$ with $\varepsilon>1$, respectively. For each $i\in \mathcal{N}$, there exists a $C^1$ function $\beta_i:\mathcal{D} \to [0,1]$ such that $\beta_i(x) = 1$, $\forall x \in \overline{\mathcal{V}}_i$, and $\beta_i(x) = 0$, $\forall x \in \mathcal{D} \setminus \mathcal{B}_i$. Since $x_i \in \operatorname{Int}(\mathcal{C})$, then $h(x_i)>0$, and we define $a_i \coloneqq \log\!\left({c_i}/{h(x_i)} \right)$, $b_i \coloneqq {w_i}/{c_i} - {\nabla h(x_i)}/{h(x_i)}$, and $s(x) \coloneqq e^{r(x)}$, where $r(x) \coloneqq \sum_{i=1}^{N} \beta_i(x) \left(a_i + b_i^\top (x - x_i)\right)$. Since $s(x)>0$, $\forall x \in \mathcal{D}$, then $\widetilde h(x)=s(x)h(x)$ has the same sign as $h(x)$ and represents the same set $\mathcal{C}$. Denote $\mathcal{P} \coloneqq \bigcup_{i\in \mathcal{N}} \overline{\mathcal{B}}_i$, since $\mathcal{P} \subset \operatorname{Int}(\mathcal{C})$, then $\mathcal{S} \coloneqq \mathcal{D} \setminus \bigcup_{i\in \mathcal{N}} \overline{\mathcal{B}}_i$ is an open neighborhood of $\partial\mathcal{C}$, so $\beta_i(x)=0$, $\forall x \in \mathcal{S}$, $\forall i \in \mathcal{N}$. Thus, we have $r(x) = 0$, $s(x) = 1$, and $\widetilde h(x)=h(x)$, $\forall x \in \mathcal{S}$. Moreover, since $x_i \in \mathcal{V}_i$ and $\beta_i(x) = 1$, $\forall x \in \mathcal{V}_i$, then we have $\beta_i(x_i) = 1$ and $\nabla{\beta}_i(x_i) = 0$. Additionally, since $x_i \notin \overline{\mathcal{B}}_j$, $\forall j \neq i \in \mathcal{N}$, and $\beta_j(x) = 0$, $\forall x \notin \mathcal{B}_j$, we have $\beta_j(x_i) = 0$ and $\nabla{\beta}_j(x_i) = 0$, $\forall j \neq i \in \mathcal{N}$. Hence, $r(x_i) = a_i$ and $\nabla{r}(x_i) = b_i$. As a result, $s(x_i) = e^{r(x_i)} = e^{a_i} = c_i/h(x_i)$, so for any $i \in \mathcal{N}$, $\widetilde{h}(x_i) = s(x_i)h(x_i) = c_i$ and $\nabla \widetilde{h}(x_i) = s(x_i) \left( \nabla h(x_i) + h(x_i) \nabla{r}(x_i) \right) = {c_i} \left(\nabla{h}(x_i) + h(x_i) \left({w_i}/{c_i} - {\nabla h(x_i)}/{h(x_i)}\right)\right)/{h(x_i)} = w_i$.

Since $h$ is assumed to be a valid CBF Definition~\ref{def:CBF}, there exists an extended class $\mathcal{K}_{\infty}$ function $\alpha$ associated with $h$. Since $\widetilde h(x)>0$, $\forall x \in \mathcal{P} \subset \operatorname{Int}(\mathcal{C})$, we denote $\bar{h} \coloneqq \min_{x\in \mathcal{P}}\widetilde h(x)>0$. Fix any $\bar{u} \in \mathcal U$. By continuity on the compact set $\mathcal{P}$, we have that $\exists \, L \in [0,\infty)$ such that $\nabla \widetilde{h}(x)^\top F(x,\bar{u}) + \alpha(\widetilde{h}(x))\geq -L$, $\forall x \in \mathcal{P}$. Choose $\lambda \geq L/\bar{h}$ and define another extended class $\mathcal K_{\infty}$ function $\widetilde{\alpha}: \mathbb{R}\to \mathbb{R}$ such that $\widetilde{\alpha}(\xi) = \alpha(\xi)$, $\forall \xi \leq 0$, and $\widetilde{\alpha}(\xi) = \alpha(\xi) + \lambda \xi$, $\forall \xi > 0$. Thus, for any $x \in \mathcal{P}$, we have $\sup_{u\in \mathcal{U}}
\left(\nabla \widetilde{h}(x)^\top F(x,u) + \widetilde{\alpha}(\widetilde h(x))\right) \geq \nabla \widetilde{h}(x)^\top F(x,\bar u) + \alpha (\widetilde h(x)) + \lambda \widetilde{h}(x) \geq -L + \lambda \bar{h} \geq 0$. We also have $\widetilde{h}(x)=h(x)$ and $\nabla \widetilde{h}(x)=\nabla{h}(x)$, $\forall x \in \mathcal{D} \setminus \mathcal{P}$. If $h(x)\leq 0$, then $\widetilde \alpha(h(x)) = \alpha(h(x))$. If $h(x)>0$, then $\widetilde \alpha(h(x)) \geq \alpha(h(x))$. Hence, the CBF validity of $h$ implies $\sup_{u \in \mathcal{U}} \left(\nabla \widetilde{h}(x)^\top F(x,u) + \widetilde{\alpha}(\widetilde{h}(x)) \right) \geq 0$, $\forall x\in\mathcal D\setminus \mathcal{P}$. Thus, $\widetilde h$ is also a valid CBF per Definition~\ref{def:CBF}. 

As such, Theorem~\ref{thm:interiorvalues} is proved.
\end{proof}

Theorem~\ref{thm:interiorvalues} shows that, if we have a valid CBF representing the safe set $\mathcal{C}$, we can construct another valid CBF representing the same set $\mathcal{C}$ by arbitrarily manipulating the CBF values and gradients at any finite collection of interior states. Therefore, the CBF values and gradients in the interior cannot define a representation-independent degree of safety, as detailed in the following corollary.
\begin{corollary}
\label{cor:nopointwisemeasure}
Consider a pointwise measure of the degree of safety $R: \mathcal{C} \times \mathcal{U} \to \mathbb{R}$ of the form
\begin{equation*}
R(x,u) = \Phi(x,u,h(x),\nabla h(x)).
\end{equation*}
If $R(x,u)$ is required to be independent of the choice of a valid CBF $h$ representing the same set $\mathcal{C}$, then $\Phi(x,u,c,w)$ must be independent of $(c,w) \in (0,\infty) \times \mathbb{R}^n$ for every $x \in \operatorname{Int}(\mathcal{C})$ and $u\in\mathcal{U}$.
\end{corollary}
\begin{proof}
Given a state $x \in \operatorname{Int}(\mathcal{C})$ and a control input $u \in \mathcal{U}$. By Theorem~\ref{thm:interiorvalues} with $N=1$, for any $c>0$ and $w\in\mathbb{R}^n$, there exists a valid CBF $\widetilde h$ representing the same set $\mathcal{C}$ such that $\widetilde{h}(x) = c$ and $\nabla\widetilde{h}(x) = w$. Representation independence requires $\Phi(x,u,h(x),\nabla{h}(x)) = \Phi(x,u,c,w)$, $\forall (c,w) \in (0,\infty) \times \mathbb{R}^n$. Since $(c,w)$ is arbitrary, $R(x,u) = \Phi(x,u,c,w)$ is constant with respect to $(c,w)$.
\end{proof}

Corollary~\ref{cor:nopointwisemeasure} shows that, if a measure of the degree of safety is required to be independent of the representation of the safe set, i.e., the particular choice of a valid CBF $h$ representing $\mathcal{C}$, then $h(x)$ and $\nabla h(x)$, $\forall x \in \operatorname{Int}(\mathcal{C})$, cannot provide any nontrivial information to such measure. Since $\dot h(x,u) = \nabla h(x)^\top F(x,u)$, Corollary~\ref{cor:nopointwisemeasure} also applies to measures involving the total time derivative of a CBF. Hence, quantities such as $\dot h$ and $\dot{h} + \alpha(h)$ remain representation-dependent in the interior despite incorporating the system dynamics. However, a meaningful measure of the degree of safety should be independent of a particular choice of CBF. 
\begin{remark}
Corollary~\ref{cor:nopointwisemeasure} has a direct consequence for CBF-guided reinforcement learning, which typically concerns the problem of $\max_{\pi} \mathbb{E}_{\pi,p} \left[\sum_{t=0}^{\infty}\gamma_{\mathrm{RL}}^{t}R(x_t,u_t)\right]$, where $x_t$ and $u_t$ are the state and control input at time step $t$, respectively, $\gamma_{\mathrm{RL}}\in[0,1)$ is a discount factor, and the expectation is taken over the policy distribution $\pi(\cdot \,|\, x_t)$, from which $u_t$ is sampled, and the state transition distribution $p(\cdot \,|\, x_t,u_t)$. For example, Corollary~\ref{cor:nopointwisemeasure} implies that the following two classes of reward functions
\begin{equation*}
\begin{aligned}
R_{\mathrm{val}}(x,u)
&=
r_{\mathrm{task}}(x,u)
+
\lambda_{\mathrm{RL}} H_{\mathrm{val}}(h(x)),\\
R_{\mathrm{dyn}}(x,u)
&=
r_{\mathrm{task}}(x,u)
+
\lambda_{\mathrm{RL}} H_{\mathrm{dyn}}
\big(h(x),\dot h(x,u)\big),
\end{aligned}
\end{equation*}
where $r_{\mathrm{task}}:\mathcal{C} \times \mathcal{U} \to \mathbb{R}$ is the task reward, $\lambda_{\mathrm{RL}}>0$ is a weighting coefficient, $H_{\mathrm{val}}:\mathbb{R} \to \mathbb{R}$, and $H_{\mathrm{dyn}}:\mathbb{R}^2 \to \mathbb{R}$, do not encode a representation-independent degree of safety. For instance, if $H_{\mathrm{val}}$ is strictly increasing, a larger reward may be assigned to a state closer to $\partial\mathcal{C}$ than to one farther away. When distance to the boundary is taken as a geometry-based degree of safety, such reward shaping may favor less safe states. Similarly, including $\dot h$ in $H_{\mathrm{dyn}}$ does not remove the representation dependence per Corollary~\ref{cor:nopointwisemeasure}. As such, considering the degree of safety can be useful when designing CBF-guided reward shaping.
\end{remark}
\subsection{Intrinsic and Representational Infeasibility} \label{sec:diagnosis}
Knowing that a state is currently in the safe set does not determine if safety can be maintained under the system dynamics and input constraints. It is therefore natural to examine if the candidate CBF-OP is feasible. However, its interpretation depends on whether the infeasibility comes from the set itself or from its CBF representation, which motivates the distinction between \emph{intrinsic infeasibility} and \emph{representational infeasibility} defined as follows.
\begin{definition}[Intrinsic and Representational Infeasibility]
\label{def:infeasibility}
Given the vector fields $f$ and $g$, the set of admissible control inputs $\mathcal{U}$, and a safe set $\mathcal C$ represented by a $C^1$ function $h$ with an extended class $\mathcal{K}_{\infty}$ function $\alpha$, intrinsic infeasibility occurs when $J_{\mathcal{C}}<0$, and representational infeasibility occurs when $J_{\mathcal C} \geq 0$ and $J_{h,\alpha} < 0$.
\end{definition}
\begin{figure}[b]
\centering
\subfigure[]{
\begin{minipage}[b]{0.234\textwidth}
\includegraphics[width=1\textwidth]{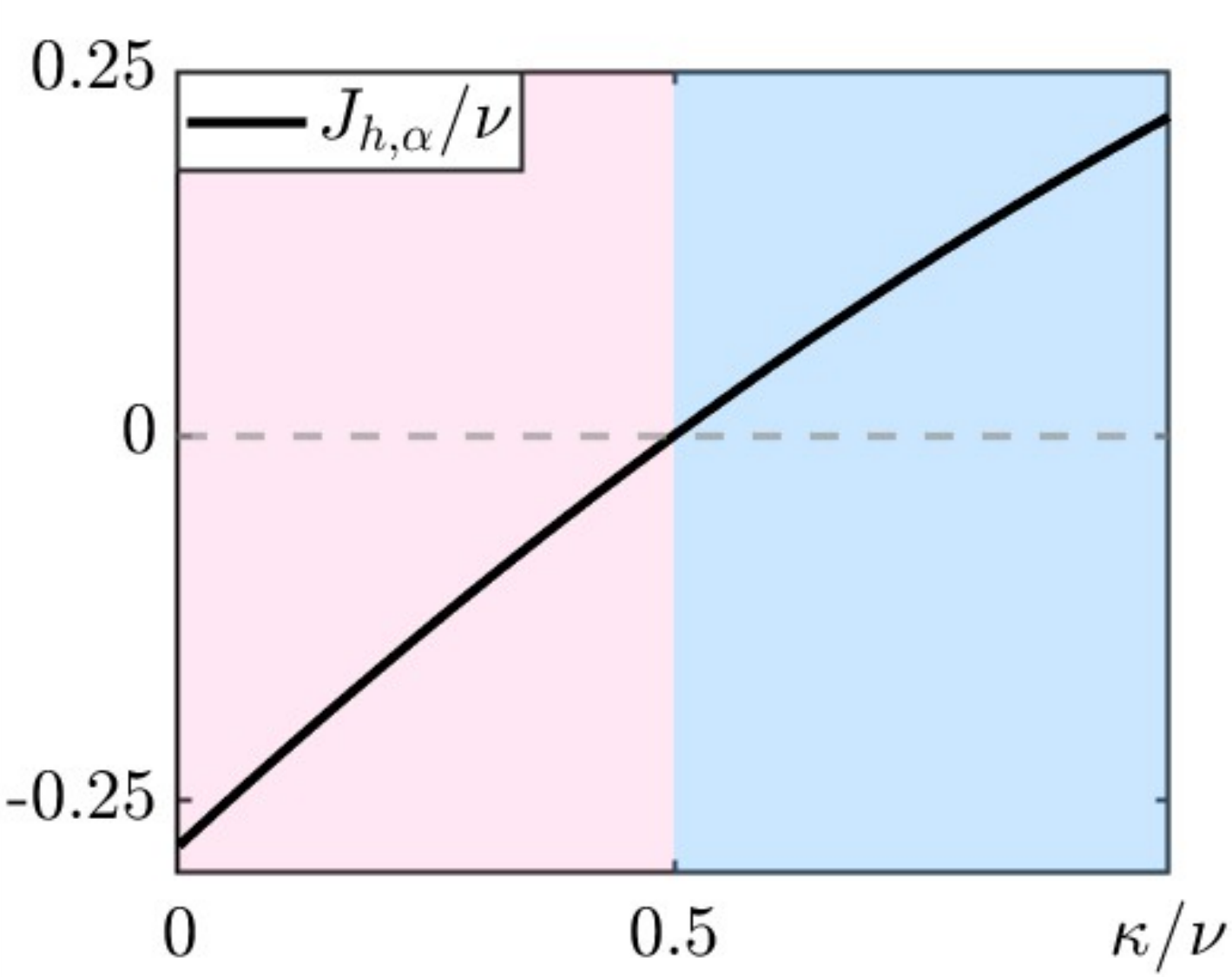}
\end{minipage}
\label{fig:repairalpha1}
}
\subfigure[]{
\begin{minipage}[b]{0.215\textwidth}
\includegraphics[width=1\textwidth]{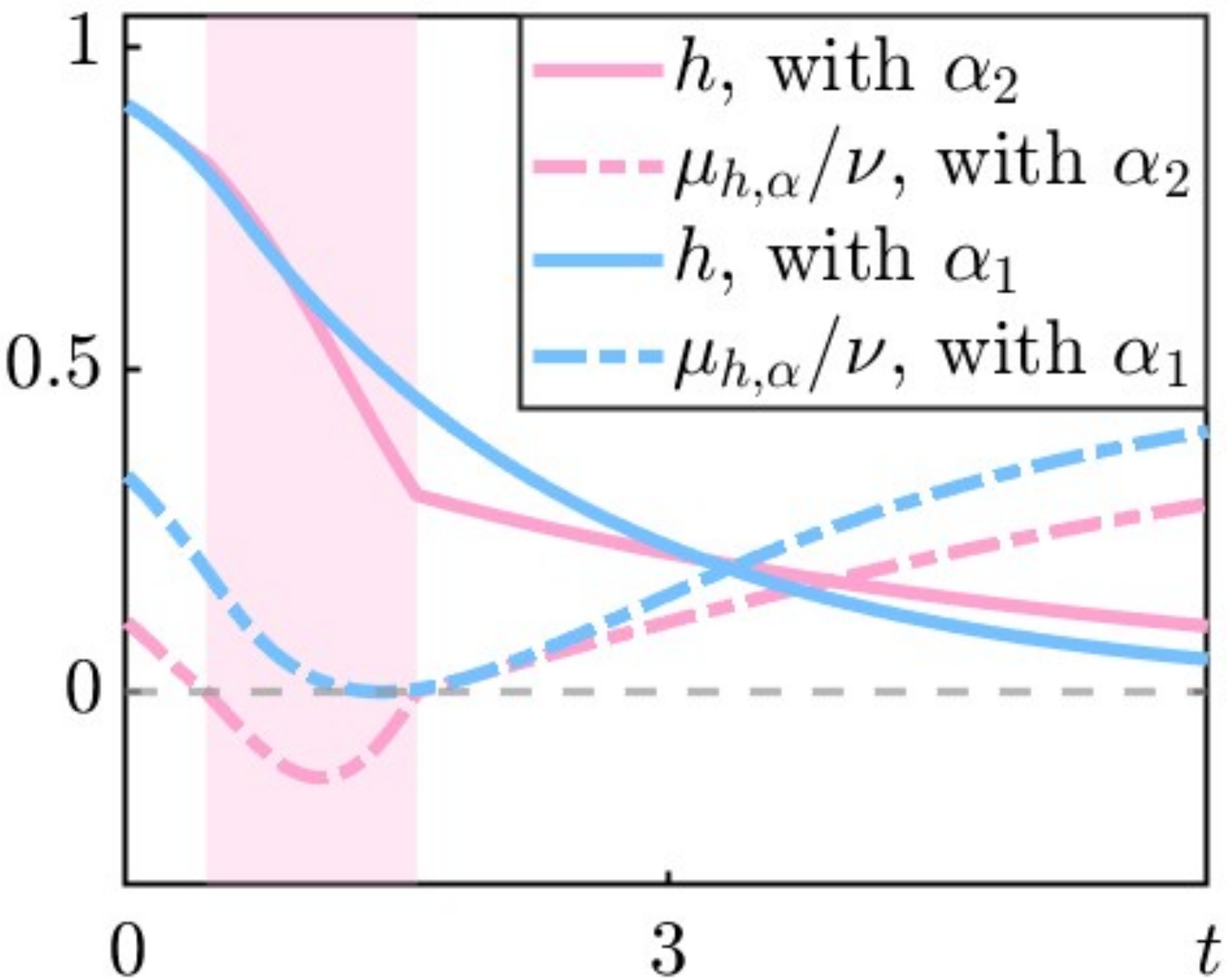}
\end{minipage}
\label{fig:repairalpha2}
}
\caption{(a) Diagnosis of representational infeasibility with the pink region indicating infeasibility and the blue region indicating feasibility; (b) Repair of representational infeasibility by changing $\alpha$, with the pink region indicating infeasibility of the corresponding OP under the bad $\alpha$.
}
\label{fig:repair_alpha}
\end{figure}

Intuitively, representational infeasibility reflects a limitation of the selected $h$ and $\alpha$, while intrinsic infeasibility reflects a limitation of the set $\mathcal{C}$ under the system dynamics $f$, $g$, and control input constraints $\mathcal{U}$.
\begin{proposition}
\label{prop:diagnosis}
\leavevmode
\begin{enumerate}[1)]
\item 
If $J_{\mathcal{C}}<0$, no $C^1$ function representing the same set $\mathcal{C}$ can be a valid CBF for any extended class $\mathcal{K}_{\infty}$ function.
\item 
If $J_{\mathcal{C}}\geq0$ but $J_{h,\alpha}<0$, the candidate CBF-based inequality constraint is infeasible somewhere in $\operatorname{Int}(\mathcal{C})$, and modifying $h$ or $\alpha$ may repair such infeasibility.
\item 
If $J_{h,\alpha}\geq0$, then $h$ is a valid CBF on $\mathcal{C}$, and any locally Lipschitz controller satisfying its CBF-based inequality renders $\mathcal{C}$ forward invariant.
\end{enumerate}
\end{proposition}
\begin{proof}
Since $\mu_{h,\alpha}(x)=\|\nabla h(x)\| \,\mu_{\mathcal C}(x)$, $\forall x \in \partial\mathcal{C}$, per Theorem~\ref{thm:Nagumo}, $J_{\mathcal{C}}<0$ implies that $\mathcal{C}$ is not controlled invariant so the infeasibility is intrinsic, and $J_{\mathcal{C}}\geq 0$ implies that $\mathcal C$ is controlled invariant so the infeasibility is representational. If $J_{h,\alpha}\geq0$, the selected $h$ is a valid CBF on $\mathcal{C}$ and the forward invariance conclusion follows from Theorem~\ref{thm:CBFinvariant}.
\end{proof}

As such, Proposition~\ref{prop:diagnosis} gives a two-stage diagnosis: The safe set should first be tested for intrinsic feasibility, and representational feasibility should be assessed only when the set is controlled invariant.

When $J_{\mathcal C}\geq0$ but $J_{h,\alpha}<0$, the safe set is controlled invariant but the representation of the set fails, i.e., the CBF-based inequality constraint becomes infeasible just because we may have picked a bad $h$ or $\alpha$, which can be repaired by modifying $h$ or $\alpha$, as shown in Example~\ref{ex:repair_representation}.
\begin{example}[Repair of Representational Infeasibility]
\label{ex:repair_representation}
Consider the dynamical system $\dot x=\nu x(1-x^2)+xu$, with $\nu>0$ and the input constraint as $\mathcal{U} = [-\nu/4,\nu/4]$. Choose the safe set as $\mathcal{C}=[-1,1]$. First, let $h(x)=1-x^2$ and the extended class $\mathcal{K}_{\infty}$ function be $\alpha(h) = \kappa h$ with $\kappa>0$, then $J_{\mathcal C}>0$, so $\mathcal C$ is controlled invariant. However, as seen from Fig.~\ref{fig:repairalpha1}, improper choice of $\kappa$ can lead to representational infeasibility. We compare $\alpha_{1}(h) = \nu h /2$ and $\alpha_{2}(h) = \nu h /4$. When $\alpha_1$ is used, $J_{h,\alpha}=0$, so $h$ in this case is a valid CBF per Definition~\ref{def:CBF}. In contrast, $J_{h,\alpha}<0$ when $\alpha_2$ is used, so the same controlled-invariant set exhibits representational infeasibility. In other words, changing $\alpha_2$ to $\alpha_1$ repairs the representational infeasibility, without changing $\mathcal{C}$ or $\mathcal{U}$. Notably, Fig.~\ref{fig:repairalpha2} also shows that even when the OP is infeasible for a period of time, the state can still stay in the set $\mathcal{C}$ during that interval, as discussed in Section~\ref{sec:intro}. Second, we fix $\alpha(h)={3\nu h}/{8}$. If we use the scalar function $h_1(x) = 1-x^2$, then $J_{h_1,\alpha} < 0$, so representational infeasibility occurs but due to the choice of $h$. However, changing $h_1 = 1-x^2$ to $h_2(x) = 1-x^4$ can repair this infeasibility because $J_{h_2,\alpha} > 0$.
\end{example}

Related approaches for addressing the issue of the infeasibility of candidate CBF-OPs include adaptive CBFs \cite{xiao2021adaptive}, optimal-decay CBFs \cite{zeng2021safety}, and rate-tunable CBFs \cite{parwana2025rate}.

When $J_{\mathcal C}<0$, changing only $h$ or $\alpha$ cannot repair the failure of safety certificate because the safe set is not controlled invariant. If $f$, $g$, and $\mathcal{U}$ are fixed, the repair must be done on the set itself, as shown in Example~\ref{ex:repair_intrinsic}.
\begin{example}[Repair of Intrinsic Infeasibility]
\label{ex:repair_intrinsic}
Consider the dynamical system $\dot{p}=v$, $\dot{v}=u$, where $p \coloneqq (p_x,p_y)^\top, v\coloneqq (v_x,v_y)^\top, u\coloneqq (u_x,u_y)^\top \in \mathbb{R}^2$, $\|u\|_\infty\leq u_{\max}$, and $0 \leq p_x\leq W$. 
As is common in CBF design for safe operating regions, we consider the collision-free workspace $\mathcal{C}_{\mathrm{geo}} = \{p \in \mathbb{R}^2 \,|\, 0\leq p_x\leq W\}$ as a position constraint on the full state $(p,v)$, and let the corresponding candidate CBFs be $h^{\ell}_{\mathrm{geo}}(p)=p_x$ and $h^r_{\mathrm{geo}}(p)=W-p_x$, which have relative degree two. Choose $\alpha(h)=\kappa h$, then the high-order CBF (HOCBF) method \cite{xiao2021high} gives two constraints $u_x + 2\kappa v_x+\kappa^2p_x \geq 0$ and $-u_x-2\kappa v_x+\kappa^2(W-p_x)\geq0$. Let $\mu^{\ell}_{\mathrm{geo}}$ denote the maximum over $\mathcal{U}$ of the left-hand side of the first constraint. We compare the trajectories obtained from the candidate HOCBF-OP with $\kappa=4$, the attempted repair that replaces $\kappa=4$ by $\tilde{\kappa}=8$ at the first loss of the OP feasibility, and the maximum braking from the beginning. Maximum braking is also applied when the corresponding OP becomes infeasible. As shown in Fig.~\ref{fig:intrinsic_repair}, the three trajectories (blue, purple, and orange) do not remain in $\mathcal{C}_{\mathrm{geo}}$ (the gray region), and changing $\kappa$ to $\tilde{\kappa}$ can only temporarily recover the OP feasibility, consistent with the fact that the constraint $p \in \mathcal{C}_{\mathrm{geo}}$ cannot be maintained from all initial states. To repair the intrinsic infeasibility, we restrict the full state to $\mathcal{C}=\{(p,v)\,|\, h^{\ell}(p,v)\geq0,\, h^r(p,v)\geq0\}$, in which $h^{\ell}(p,v)=p_x-{\max\{-v_x,0\}^2}/{2u_{\max}}$ and $h^r(p,v) = W-p_x-{\max\{v_x,0\}^2}/{2u_{\max}}$. Such $\mathcal{C}$ is controlled invariant, on which $h^{\ell}(p,v)$ and $h^r(p,v)$ are valid CBFs.
\end{example}

Example~\ref{ex:repair_intrinsic} highlights the value of the diagnosis in Proposition~\ref{prop:diagnosis}. When infeasibility occurs, before attempting to manipulate $h$ or $\alpha$, one should first determine whether the safe set is intrinsically feasible. If the safe set already contains unrecoverable states, then adjusting only the representation of the set is the wrong direction.
\begin{figure}[t]
\centering
\subfigure[]{
\begin{minipage}[b]{0.2\textwidth}
\includegraphics[width=1\textwidth]{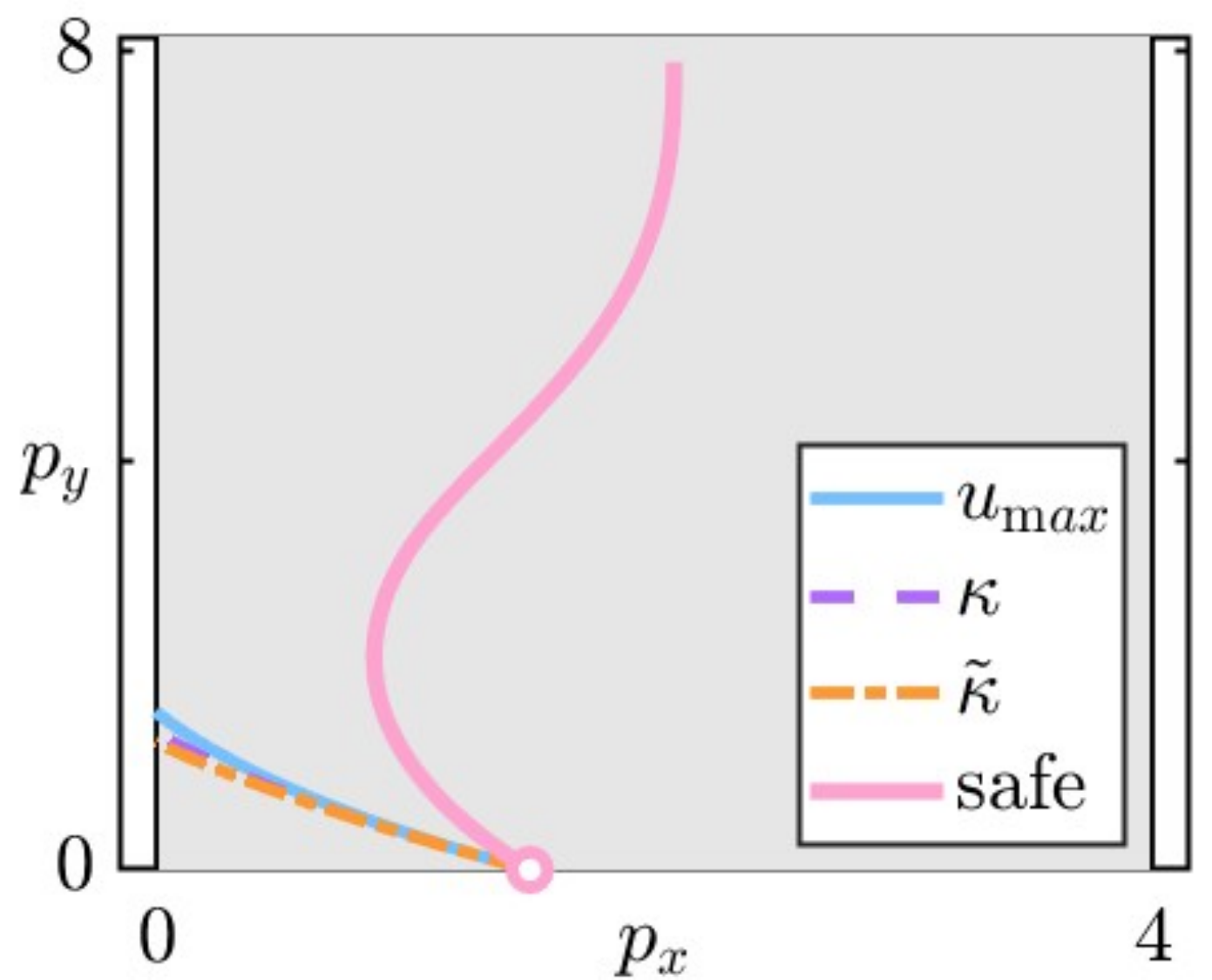}
\end{minipage}
\label{fig:intrinsic_workspace}
}
\subfigure[]{
\begin{minipage}[b]{0.2\textwidth}
\includegraphics[width=1\textwidth]{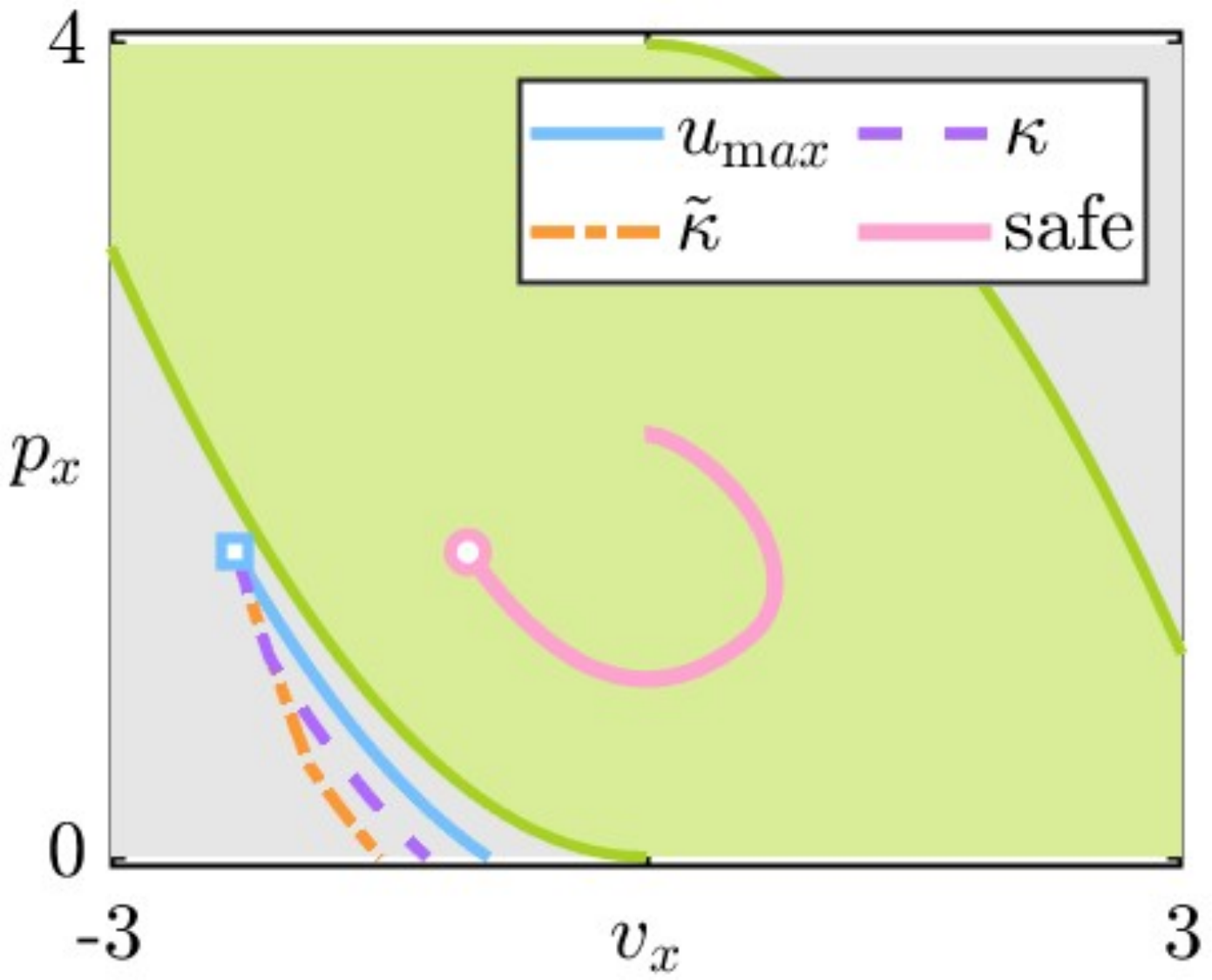}
\end{minipage}
\label{fig:intrinsic_state}
}
\subfigure[]{
\begin{minipage}[b]{0.205\textwidth}
\includegraphics[width=1\textwidth]{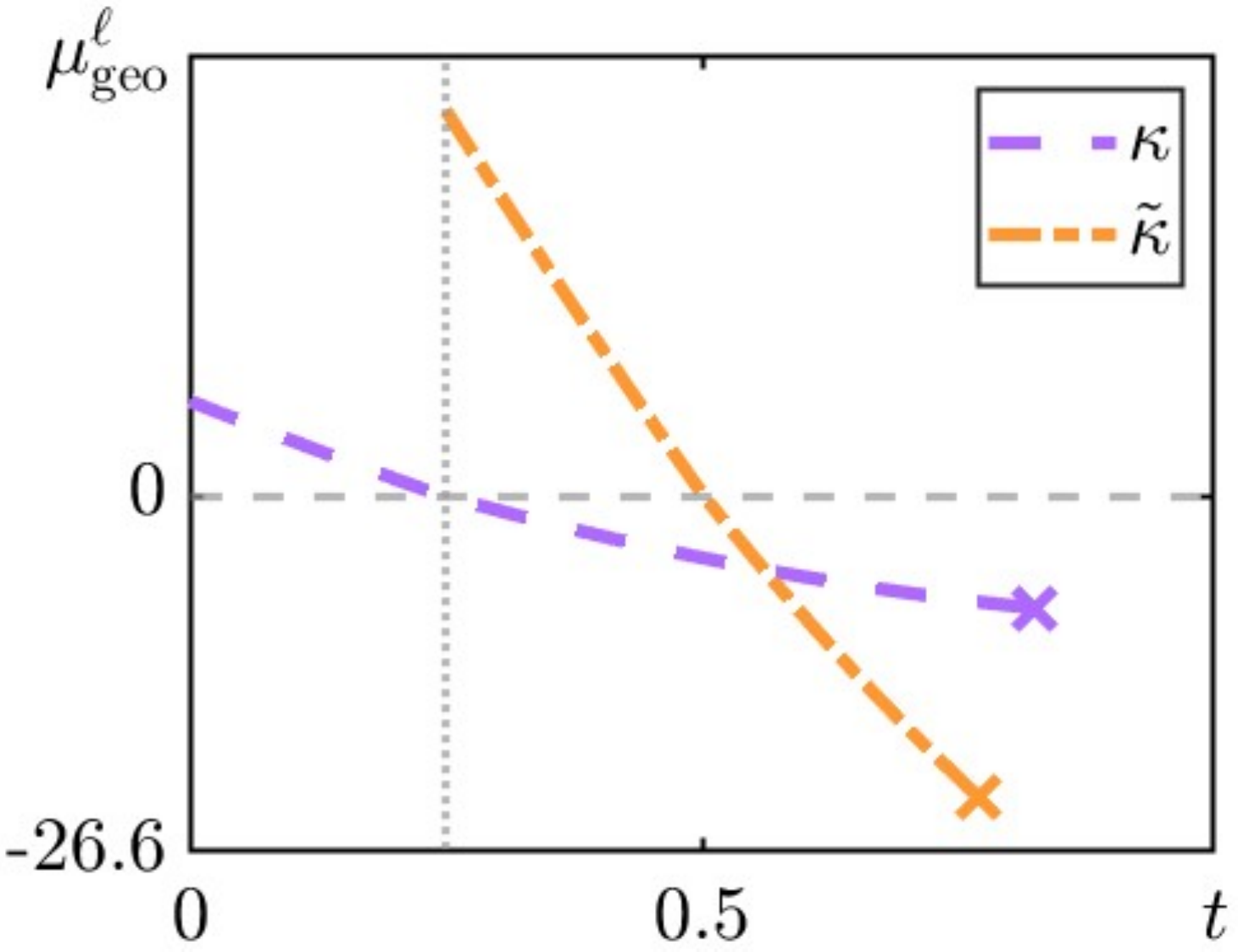}
\end{minipage}
\label{fig:intrinsic_kappa}
}
\subfigure[]{
\begin{minipage}[b]{0.2\textwidth}
\includegraphics[width=1\textwidth]{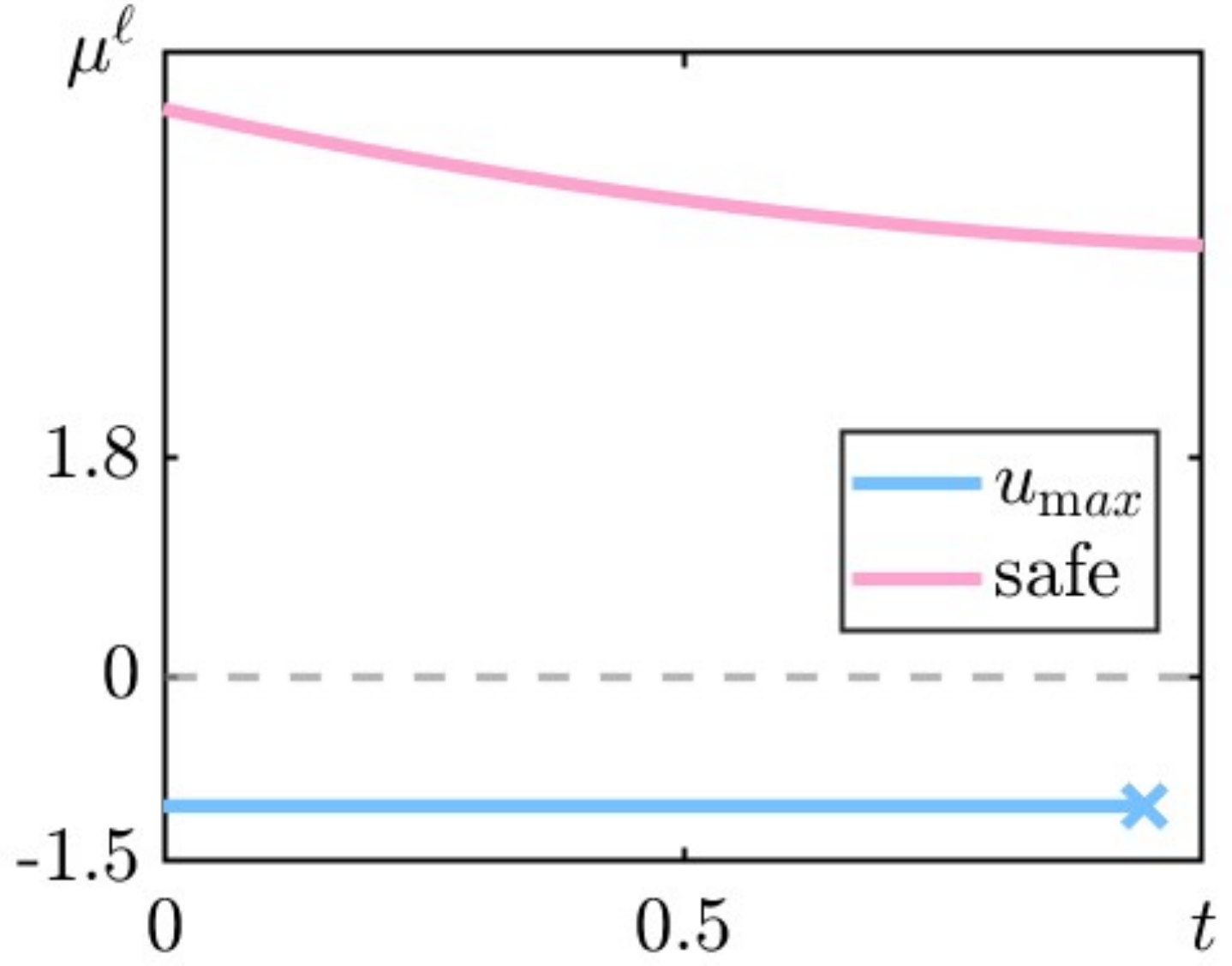}
\end{minipage}
\label{fig:intrinsic_mu}
}
\caption{(a) Trajectories in the workspace under the fixed $\kappa$, repaired $\tilde{\kappa}$, maximum braking, and safe controller; (b) The same trajectories in the $(v_x,p_x)$ space, where the green set is controlled invariant; (c) Loss of candidate HOCBF-OP feasibility with $h^{\ell}_{\mathrm{geo}}(p)=p_x$ and temporary recovery after tuning $\kappa$; (d) Intrinsic and representational feasibility (pink) under safe controller and intrinsic infeasibility even under maximum braking (blue), where $\mu^{\ell}$ is per \eqref{eqn:mu_halpha} with $h^{\ell}$. A cross indicates when the corresponding trajectory first crosses the boundary of $\mathcal{C}_{\mathrm{geo}}$.}
\label{fig:intrinsic_repair}
\end{figure}

Related approaches for obtaining controlled invariant sets include control barrier-value functions and Hamilton-Jacobi refinement of candidate CBFs \cite{choi2021robust,tonkens2022refining}, convex computation of maximum controlled invariant sets \cite{korda2014convex}, and input-constrained or backup CBF constructions \cite{agrawal2021safe,chen2021backup}.

In general, verifying intrinsic or representational infeasibility may be as computationally challenging as verifying whether a scalar function is a valid CBF per Definition~\ref{def:CBF}. Our purpose is to identify both the source of infeasibility and the appropriate repair direction.
\subsection{A Control Authority-Based Degree of Safety} \label{sec:degrees}
The results so far exclude the use of CBF values, gradients, and candidate CBF-OP feasibility to quantify the degree of safety in a representation-independent manner, since they all depend on the specific choice of CBF. However, $J_{\mathcal{C}}$, which is evaluated on $\partial \mathcal{C}$, tests the controlled invariance without a specific CBF representation of $\mathcal{C}$. This motivates a control authority-based degree of safety defined below.
\begin{definition}[Invariance Authority Demand (IAD)]
\label{def:actuationdemand}
For any $x\in\partial\mathcal{C}$, the pointwise IAD is defined as
\begin{equation}
\gamma_{\mathcal{C}}(x)
\coloneqq
\inf
\left\{
\rho\geq0
\;\bigg|\;
\sup_{u \in \rho \mathcal{U}}
n_{\mathcal{C}}(x)^\top F(x,u) \geq 0
\right\},
\label{eqn:actuationdemand}
\end{equation}
where $\rho \mathcal{U} \coloneqq \{\rho u \,|\, u\in\mathcal{U}\}$ and $0 \in \mathcal{U}$. The setwise IAD is defined as
\begin{equation}
\Gamma_{\mathcal{C}}
\coloneqq
\sup_{x\in\partial\mathcal{C}}
\gamma_{\mathcal{C}}(x).
\label{eqn:Gamma}
\end{equation}
\end{definition}

The pointwise IAD $\gamma_{\mathcal{C}}(x)$ defined in \eqref{eqn:actuationdemand} is dimensionless and measures the control authority demand, relative to the full actuator capability, required to prevent the system from instantaneously leaving the safe set $\mathcal{C}$ at $x \in \partial \mathcal{C}$. Intuitively, the geometries of $\mathcal{U}$ and $\mathcal{C}$ and the control vector field $g$ jointly determine the most effective admissible control direction for preventing the state from leaving $\mathcal{C}$, and the drift vector field $f$ determines how much control effort is required along such a direction.
\begin{proposition}
\label{prop:actuationdemand}
The pointwise IAD in \eqref{eqn:actuationdemand} satisfies
\begin{equation}
\gamma_{\mathcal{C}}(x)
=
\begin{cases}
0,
& \delta_f(x)\geq0,\\
-\dfrac{\delta_f(x)}{\delta_u(x)},
& \delta_f(x)<0,\; \delta_u(x)>0,\\
+\infty,
& \delta_f(x)<0,\; \delta_u(x)=0,
\end{cases}
\label{eqn:actuationclosed}
\end{equation}
in which $\delta_u(x) \coloneqq \sup_{u\in\mathcal{U}} n_{\mathcal{C}}(x)^\top g(x)u$ and $\delta_f(x) \coloneqq n_{\mathcal{C}}(x)^\top f(x)$. In addition,
\begin{equation}
\gamma_{\mathcal{C}}(x)\leq 1
\;\Longleftrightarrow\;
\mu_{\mathcal{C}}(x)\geq 0,
\label{eqn:actuationequivalence}
\end{equation}
\begin{equation}
\mathcal{C}\text{ is controlled invariant}
\;\Longleftrightarrow\;
\Gamma_{\mathcal{C}}\leq1.
\label{eqn:Gammainvariance}
\end{equation}
\end{proposition}
\begin{proof}
Since $\sup_{u\in\rho \mathcal{U}} n_{\mathcal{C}}(x)^\top F(x,u) = \delta_f(x)+\rho\delta_u(x)$, solving $\delta_f(x) + \rho \delta_u(x) \geq 0$ gives \eqref{eqn:actuationclosed}. Letting $\rho=1$ gives $\mu_{\mathcal{C}}(x)=\delta_f(x)+\delta_u(x)$, which, together with \eqref{eqn:actuationclosed}, results in \eqref{eqn:actuationequivalence}. Taking the supremum over $\partial\mathcal{C}$ and applying Theorem~\ref{thm:Nagumo} leads to \eqref{eqn:Gammainvariance}.
\end{proof}

Both the pointwise IAD $\gamma_{\mathcal{C}}$ in \eqref{eqn:actuationdemand} and the setwise IAD $\Gamma_{\mathcal{C}}$ defined in \eqref{eqn:Gamma} depend only on the safe set $\mathcal{C}$, system dynamics $f$ and $g$, and the set of admissible control inputs $\mathcal{U}$. Thus, they are unchanged when $\mathcal{C}$ is represented by a different candidate CBF, i.e., they are both representation-independent. In addition, the geometry of $\mathcal{U}$ automatically accounts for how control authority is measured, as \eqref{eqn:actuationdemand} already encodes directional actuator capability. In particular, if
\begin{equation*}
\mathcal{U}
=
\left\{
D u \mid \|u\|_p\leq 1
\right\},
\end{equation*}
where $D\in \mathbb{R}^{m\times m}$ and $p\in [1,\infty]$, then
\begin{equation*}
\delta_u(x)
=
\left\|
D^\top g(x)^\top n_{\mathcal{C}}(x)
\right\|_q,
\end{equation*}
where $q\in[1,\infty]$ satisfies ${1}/{p}+{1}/{q}=1$. For example, if $D=\operatorname{diag}(d_1,\ldots,d_m)$ with $d_i>0$, $\forall i = 1,\ldots,m$, then a box-shaped $\mathcal{U}$ induces a weighted $L_1$ norm in $\delta_u(x)$, an ellipsoidal $\mathcal{U}$ induces a weighted $L_2$ norm in $\delta_u(x)$, and a diamond-shaped $\mathcal{U}$ induces a weighted $L_\infty$ norm in $\delta_u(x)$.

The setwise IAD $\Gamma_{\mathcal{C}}$ \eqref{eqn:Gamma} quantifies the control authority required to maintain safety, and can be used to guide the repair of intrinsic infeasibility by reshaping the set or redesigning the actuator. Specifically, $\Gamma_{\mathcal{C}} \leq 1$ indicates controlled invariance, and a smaller $\Gamma_{\mathcal{C}}$ indicates less required control authority, relative to the full actuator capability, to maintain safety, i.e., a higher degree of safety. In addition, when $\Gamma_{\mathcal C}<1$, controlled invariance is maintained under any reduced input set $\rho\mathcal U$ with $\Gamma_{\mathcal C}\leq\rho<1$, which can inform the choice of input constraints in safety-critical controller design. When $\Gamma_{\mathcal{C}} > 1$, its magnitude quantifies how much additional control authority is needed to make $\mathcal{C}$ controlled invariant, as summarized in the following corollary.
\begin{corollary}
\label{cor:authorityscaling}
If $\Gamma_{\mathcal{C}} < +\infty$, then, for any $\rho \geq 0$,
\begin{equation*}
\mathcal{C}\text{ is controlled invariant under }
u\in\rho\mathcal{U}
\;\Longleftrightarrow\;
\rho\geq\Gamma_{\mathcal{C}}.
\end{equation*}
\end{corollary}
\begin{proof}
By Theorem~\ref{thm:Nagumo} and Definition~\ref{def:actuationdemand}, $\mathcal{C}$ is controlled invariant under $\rho \mathcal{U}$ when $\rho\geq\gamma_{\mathcal{C}}(x)$, $\forall x\in\partial\mathcal{C}$, which is equivalent to
$\rho\geq\Gamma_{\mathcal{C}}$ based on \eqref{eqn:Gamma}.
\end{proof}

Corollary~\ref{cor:authorityscaling} directly provides two directions for repairing intrinsic infeasibility of $\mathcal{C}$ under $f$, $g$, and $\mathcal{U}$: (i) If the safe set $\mathcal{C}$ needs to be unchanged, the available control authority must be increased, and $\Gamma_{\mathcal{C}}$ gives exactly the minimum scaling of $\mathcal{U}$ required to make $\mathcal{C}$ controlled invariant; (ii) If the actuator needs to be unchanged, the repair must be performed on the safe set itself, motivating the set repair method below.

\begin{figure}[t]
\centering
\subfigure[]{
\begin{minipage}[b]{0.2\textwidth}
\includegraphics[width=1\textwidth]{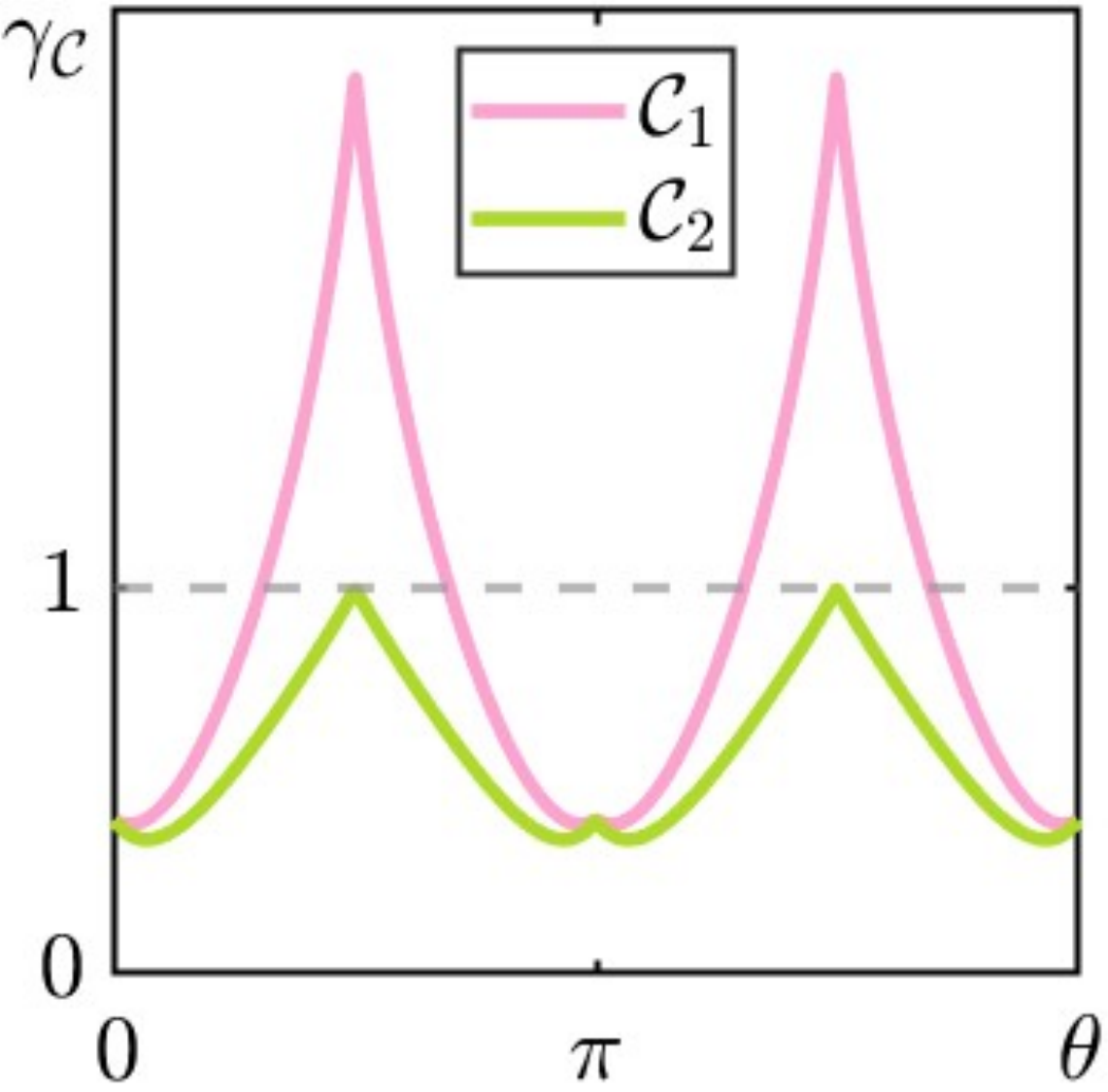}
\end{minipage}
}
\subfigure[]{
\begin{minipage}[b]{0.2\textwidth}
\includegraphics[width=1\textwidth]{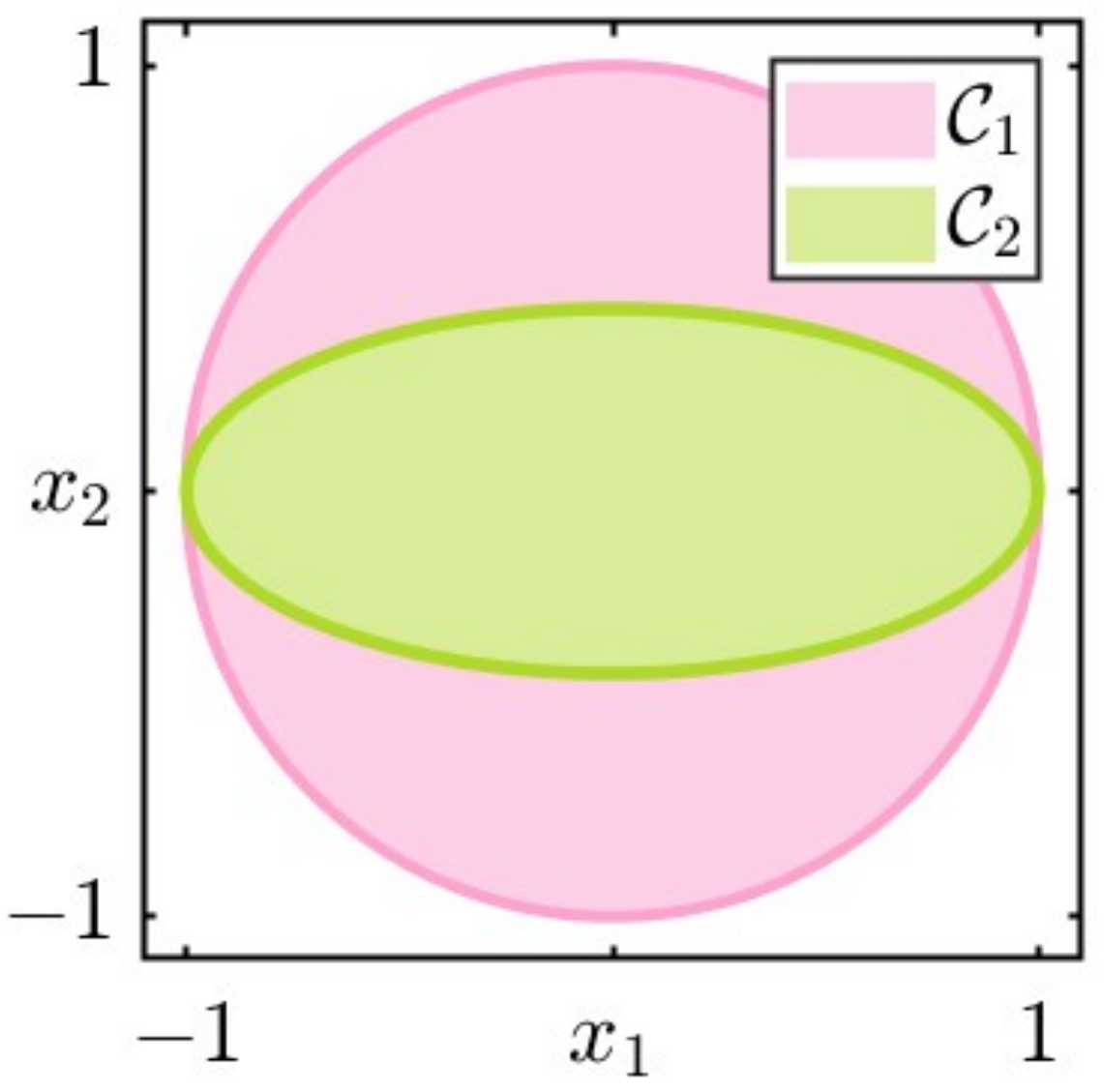}
\end{minipage}
}
\caption{IAD-guided repair from the intrinsically infeasible set $\mathcal{C}_1$ to the controlled invariant set $\mathcal{C}_2$.}
\label{fig:actuationguidedrepair}
\end{figure}

Given a set $\mathcal{A} \subseteq \mathcal{D}$, and a parameterized set family $\{\mathcal{C}_{\theta} \,|\, \theta \in \Theta\}$, let $\mathcal{Q}(\mathcal{C}_{\theta}) \in \mathbb{R}$ denote a measure depending on a specific task, such as volume. One may consider
\begin{equation}
\begin{aligned}
\sup_{\theta\in\Theta}\quad
&\mathcal{Q}(\mathcal{C}_{\theta})\\
\text{s.t.}\quad
&\mathcal{C}_{\theta}\subseteq\mathcal{A},\\
&\Gamma_{\mathcal{C}_{\theta}}\leq 1.
\end{aligned}
\label{eqn:setselection}
\end{equation}
\begin{example}
\label{ex:setrepair}
Consider the dynamical system $\dot{x} = \operatorname{diag}(\eta_1,\ldots,\eta_n) x + \operatorname{diag}(\zeta_1,\ldots,\zeta_n) u$ with $\eta_i \geq 0$, $\zeta_i > 0$, and $\mathcal{U} = [-1,1]^n$, and the set family $\mathcal{C}_{r} = \left\{x\in\mathbb{R}^n \,\big|\, \sum_{i=1}^{n}{x_i^2}/{r_i^2}\leq 1 \right\}$ with $r_i>0$. We have $\gamma_{\mathcal{C}_{r}}=\left(\sum_{i=1}^{n}\eta_i\omega_i^2\right) / \left(\sum_{i=1}^{n}{\zeta_i}|\omega_i|/{r_i}\right)$, where $\omega_i=x_i/r_i$, and thus $\Gamma_{\mathcal{C}_{r}} = \max_i \left({\eta_i r_i}/{\zeta_i}\right)$.
\end{example}
For Example~\ref{ex:setrepair}, with fixed maximum actuator capability, the safe set should be reduced more along directions with stronger outward drift and/or weaker control authority, while directions with more remaining control authority can be preserved or even enlarged according to the task objective. Applying \eqref{eqn:setselection} to a special case of Example~\ref{ex:setrepair} with $\eta_1 = 0.4$, $\eta_2 = 1.4$, $\zeta_1 = 1$, and $\zeta_2 = 0.6$ results in the maximum volume repair of the intrinsically infeasible set $\mathcal{C}_1$ to the controlled invariant set $\mathcal{C}_2$, as shown in Fig.~\ref{fig:actuationguidedrepair}, where the horizontal direction (with larger control authority) is preserved, and the vertical direction (with smaller control authority) is contracted.

\section{Conclusion} \label{sec:conclusion}
Beyond a binary safety statement, what should ``safer'' mean in safety-critical control? We show that CBF values, CBF gradients, and candidate CBF-OP feasibility do not by themselves provide such a quantitative measure. The distinction between intrinsic and representational infeasibility further clarifies why even the OP infeasibility must be interpreted carefully. In general, there is no unique measure of the degree of safety, and the appropriate choice depends on specific objectives or requirements. Distance to the set boundary and the proposed IAD are two examples capturing different aspects of safety degrees, both independent of the CBF representation of the safe set. We hope this work encourages further study of the degree of safety.

\bibliographystyle{IEEEtran}
\bibliography{references}
\end{document}